\documentclass[submission,copyright,creativecommons]{eptcs}
\providecommand{\event}{QPL 2026} 

\usepackage{iftex}

\ifpdf
  \usepackage{underscore}         
  \usepackage[T1]{fontenc}        
\else
  \usepackage{breakurl}           
\fi

\title{Quantum Theory over Dual-Complex Numbers}
\author{Nathan Houyet
\institute{Université Paris-Saclay, Inria, CNRS, LMF\\ 91190 Gif-sur-Yvette, France}
\email{nathan.houyet@ens-paris-saclay.fr}
\and
Dogukan Bakircioglu
\institute{Université Paris-Saclay, Inria, CNRS, LMF\\ 91190 Gif-sur-Yvette, France}
\email{dogukan.bakircioglu@inria.fr}
\and
Pablo Arrighi
\institute{Université Paris-Saclay, Inria, CNRS, LMF\\ 91190 Gif-sur-Yvette, France}
\email{pablo.arrighi@inria.fr}
}
\newcommand{\titlerunning}{Quantum Theory over Dual-Complex Numbers}
\newcommand{\authorrunning}{N.L. Houyet, D. Bakircioglu \& P. Arrighi}

\hypersetup{
  bookmarksnumbered,
  pdftitle    = {\titlerunning},
  pdfauthor   = {\authorrunning},
  pdfsubject  = {EPTCS},               
}

\usepackage[english]{babel}
\usepackage[utf8]{inputenc}
\usepackage{amssymb}
\usepackage{amsmath}
\usepackage{amsthm}
\usepackage{environ}
\usepackage{thmtools}
\usepackage{braket}
\usepackage{multicol}
\usepackage{graphicx}
\usepackage{subcaption}
\graphicspath{{REFS/LorentzDirac/}}
\usepackage{comment}

\usepackage{hyperref}
\usepackage{pgf}
\usepackage{tikz}
\usetikzlibrary{arrows,automata,positioning}
\newsavebox{\gates}
\newsavebox{\wires}
\newsavebox{\crossings}
\newsavebox{\patchuno}
\newsavebox{\patchdos}
\newsavebox{\patchtres}
\newsavebox{\patchcuatro}
\newsavebox{\patchcinco}
\newsavebox{\patchseis}
\newcommand{\ii}{i}
\newcommand{\crossingswname}[7]{%
\foreach \i in {0,...,\the\numexpr #2-1 \relax}{%
  \draw[#5,thick] (0, \i+1) node[black,#6] {#3}  -- (#1+0, \i+1)  ;
  \draw[#5,thick] (#1+0,\i+1) -- (#1+1.5,\i+1);
};
\foreach \i in {0,...,\the\numexpr #1 - 1 \relax}{%
  \draw[#5,thick] (\i+1, 0) node[black,#7] {#4} -- (\i+1, #2+0);
  \draw[#5,thick] (\i+1,#2+0) -- (\i+1,#2+1.5);
};
\foreach \i in {1,...,#1}{%
  \foreach \j in {1,...,#2}{%
    \draw[fill=white] (\i,\j) circle (0.1);
}};
}

\newtheorem{theorem}{Theorem}

\newtheorem{definition}[theorem]{Definition}
\newtheorem{proposition}[theorem]{Proposition}
\newtheorem{corollary}[theorem]{Corollary}
\newtheorem{example}[theorem]{Example}
\newtheorem{postulate}{Postulate}

\newif\iftoappendix
\toappendixtrue

\NewEnviron{proofcontent}[1]{%
  \iftoappendix
    \expandafter\ifx\csname savedproof_#1_toks\endcsname\relax
      \expandafter\newtoks\csname savedproof_#1_toks\endcsname
    \fi
    \expandafter\global\expandafter\csname savedproof_#1_toks\endcsname\expandafter=\expandafter{\BODY}%
    \par\noindent\textit{Proof in App.~\ref{app:proof:#1}.}%
  \else
    \begin{proof}\BODY\end{proof}%
  \fi
}

\newcommand{\e}{\varepsilon}
\newcommand{\N}{\mathbb{N}}
\newcommand{\R}{\mathbb{R}}
\newcommand{\C}{\mathbb{C}}
\newcommand{\D}{\mathbb{D}}
\newcommand{\DC}{\mathbb{DC}}
\newcommand{\Z}{\mathcal{O}_\e}

\newcommand{\E}{\mathcal{E}}

\newcommand{\til}{\widetilde}
\renewcommand{\bar}{\overline}

\renewcommand{\Re}{\operatorname{Re}}
\renewcommand{\Im}{\operatorname{Im}}
\newcommand{\Sig}{\operatorname{Sig}}
\newcommand{\Inf}{\operatorname{Inf}}

\newcommand{\dual}{dual-complex }
\newcommand{\Dual}{Dual-complex }

\newcommand\dstate[2]{\frac{\mathrm{d}\ket{#1}}{\mathrm{d}#2}}

\newcommand\ketbra[2]{\ket{#1}\!\bra{#2}}

\newcommand{\norm}[1]{\left\lVert#1\right\rVert}

\begin{document}
\maketitle

\begin{abstract}
We take quantum theory and replace $\mathbb{C}$ by $\mathbb{C}[\varepsilon]$ where $\e^2=0$, i.e. we extend quantum theory to the ring of dual complex numbers. The aim is to develop a common language in which to treat continuous quantum physics and discrete quantum models in a unified manner, including their symmetries. Since quantum theory is linear, introducing $\e$ is enough to model infinitesimals. A first objection to this programme is that $\mathbb{C}[\varepsilon]$ is not a field, since division by $\varepsilon$ is undefined, while quantum mechanics typically relies on division. A second objection concerns whether unitarity still makes sense given $\varepsilon^2 = 0$. Hence, the core of this work is dedicated to proving that \dual quantum theory remains fully consistent. In particular, norm is preserved at all times, and renormalization never requires dividing by an infinitesimal. An equivalence with conventional quantum theory is demonstrated: the \dual extension of a parametrized quantum operation automatically provides a linear treatment of its first-order variations. As a first example application, we provide a unified description of both the Dirac equation in the continuum and the Dirac Quantum Walk in the discrete. We establish the discrete Lorentz covariance of the latter.
\end{abstract}

\section{Introduction}

We replace the complex field by the \dual ring in the postulates of quantum mechanics. This modification allows us to linearly approximate the action of families of quantum operations $\E_h(\rho)$ where $h$ is a real parameter. These numbers are of the form $z + t \e$ where $\e^2 = 0$ and $z$ and $t$ are complex numbers. Intuitively $t\e$ symbolizes a small variation of $z$. \Dual numbers have found wide applications in automatic differentiation, mechanics and geometry \cite{brodsky1999, baydin2018, yang1964, rev2016}. Recent works have begun to extend our knowledge of \dual linear operators \cite{qi2024,liu2025}. An alternative synthetic treatment of infinities in quantum theory, based on non-standard analysis rather than nilpotent dual numbers, was developed in~\cite{gogioso2016}.\\

In this paper, we build upon these developments by formulating a self-consistent set of postulates that extend conventional quantum theory into the realm of \dual numbers. We introduce translation rules that automate the mapping between the \dual and the conventional formalisms of quantum mechanics. Some of our results extend what was previously known about \dual unitary operators. These technical results may be of independent mathematical interest.\\

We argue that our extension is physical and offers a rigorous framework for a widely used but often informal technique in physics—the neglect of higher-order terms in functional expansions. Since $\e^2=0$, the algebra of \dual numbers encodes first-order approximations, making it well-suited for modelling systems where higher-order effects should be neglected. We provide an example of application of our framework where it allows for an easy derivation of Dirac equation from a Discrete time Quantum Walk (QW).\\

Beyond their technical utility, \dual numbers also provide a compelling perspective in foundational debates, particularly the tension between discrete and continuous models of physics. While discrete models are computationally convenient and visually intuitive, continuous models have been more studied and are sometimes easier to work with, especially when symmetries are concerned for instance. This raises the question whether discrete models can fully represent Nature \cite{kornyak2013, hossenfelder2012, stecker2011, chamseddine2014, chamseddine2018}.
Through our example, we show that the \dual framework has the potential to unify these two perspectives by combining a continuous base field $\C$ with an infinitesimal variation which represents a discrete step. Such a hybrid structure may facilitate smooth translations between continuous and discrete frameworks, potentially offering new tools for quantum computation, simulation and foundations.\\

\Dual numbers constitute a new option in the foundational debate about which set of numbers is needed or useful to express quantum behaviours. Quantum theory based on extension of complex numbers such as quaternion numbers has been known for a long time (\cite{finkelstein1962}) while the debate on the necessity of complex numbers over real numbers is still raging \cite{hoffreumon2025}.\\

In Sec.\ref{sec:dualnumbers}, we recall all the required notions about \dual numbers for later sections. In Sec.\ref{sec:linearalgebra}, we recall a few results about \dual linear operators and extend what is known about \dual unitary operators. We also define a notion of ordering required for our self-consistency check. In Sec.\ref{sec:postulates}, we use these results to present a set of postulates whose consistency is proved in App.~\ref{sec:consistency}.
Finally, in Sec.\ref{sec:QCA}, we present how \dual numbers are employed to automate the continuum limit calculations of the Dirac QW \cite{ArrighiDirac,d2014derivation}, and make its discrete Lorentz covariance \cite{ArrighiLorentzCovariance} exact.

\section{Dual numbers}\label{sec:dualnumbers}

\noindent Extending real numbers with an imaginary unit $i$ with $i^2 = -1$ gives rise to the complex numbers. Complex numbers form a field. This field can be used to define a vector space. Quantum mechanics represents physical systems and their evolution within that vector space.

{\em Notations.} As usual the set of complex numbers $\{z = a + bi \; | \; a, b \in \R \text{ and } {i^2 = -1}\}$ is denoted $\C$. As usual, we denote by real part ($\Re$) and imaginary part ($\Im$) the functions:
\begin{multicols}{3}
\noindent
\begin{equation}
\Re: \C \to \R: a + bi \to a
\end{equation}

\columnbreak
\noindent
\begin{equation}
\Im: \C \to \R: a + bi \to b
\end{equation}

\columnbreak
\noindent
\begin{equation}
\forall z \in \C, z = \Re(z) + \Im(z)i
\end{equation}

\end{multicols}

Dual numbers $\D$ were first introduced in \cite{clifford1871} and follow a similar idea. They are generated by extending the real numbers with an {\em infinitesimal unit} $\e$ such that $\e^2 = 0$. Multiplication and addition remain associative and commutative. Analogously, given a dual number $d\in\D$, we can take its {\em significant part} ($\Sig$) and its {\em infinitesimal part} ($\Inf$).

\begin{multicols}{3}
\noindent
\begin{equation}
\Sig: \D \to \R: a + b\e \to a
\end{equation}

\columnbreak
\noindent
\begin{equation}
\Inf: \D \to \R: a + b\e \to b
\end{equation}

\columnbreak
\noindent
\begin{equation}
\forall d \in \D, d = \Sig(d) + \Inf(d)\e
\end{equation}

\end{multicols}

Extending the real numbers with both $i$ and $\e$ gives rise to \dual numbers $\DC$. We can understand the complex, dual and \dual numbers as quotient rings:

\begin{multicols}{3}

\noindent
\begin{equation}
\C \approx \R [i]/\langle i^2+1 \rangle
\end{equation}

\columnbreak

\noindent
\begin{equation}
\D \approx \R [\e]/\langle \e^2 \rangle
\end{equation}

\columnbreak

\noindent
\begin{equation}
\DC \approx \C [\e]/\langle \e^2 \rangle.
\end{equation}
\end{multicols}

Alternatively, \dual (resp. dual) numbers can be seen as the subalgebra of the matrices of the form

\begin{equation}
\begin{pmatrix}
 a & b\\
 0 & a
\end{pmatrix},
\end{equation}

where $a, b \in \C$ (resp. $a, b \in \R$). Dual and \dual numbers are widely used in automatic differentiation as they allow to ``cut'' the Taylor expansion and analytic functions from the second term on \cite{baydin2018, rev2016}.

\begin{equation}
f(z + \e) = f(z) + \e f'(z) + \e^2 f''(z)/2 + \dots = f(z) + \e f'(z).
\end{equation}

Higher-order versions of dual-numbers were developed to deal with higher-order derivative, see \cite{szi2021, behr2019}.

\subsection{Ring operations and inversion on \dual numbers}

Let $ z_{i},t_{i} \in \C, \forall i \in \N $ and $ z_{i}+ t_{i} \e \in \DC$. The definitions of addition and multiplication follow directly from associativity, commutativity, and $\e^2=0$:

\begin{equation}
(z_1 + t_1 \e) + (z_2 + t_2 \e) := (z_1 + z_2) + (t_1 + t_2) \e
\end{equation}
\noindent \begin{equation}
(z_1 + t_1 \e) (z_2 + t_2 \e) := (z_1 z_2) + (t_1 z_2 + z_1 t_2) \e .
\end{equation}

A \dual is called {\em infinitesimal} when its significant part is zero;  the set of infinitesimal \dual numbers is denoted $\Z = \{0 + t \e | t \in \C\} = \e \C$. A number that is not infinitesimal is called {\em appreciable}. When dividing $w_1 = z_1 + t_1 \e$ by $w_2 = z_2 + t_2 \e$ there are three cases:

\begin{enumerate}
        \item When $w_2$ is appreciable, i.e.\ $z_2 \neq 0$. In this case, division is defined as
        \begin{equation}\label{eq:inverse}
        \frac{z_1 + t_1 \e}{z_2 + t_2 \e} := \frac{(z_1 + t_1 \e)(z_2 - t_2 \e)}{(z_2 + t_2 \e)(z_2 - t_2 \e)} = \frac{z_1 z_2 - z_1 t_2 \e + t_1 z_2 \e}{z_2^2} = \frac{z_1}{z_2} + \frac{z_2 t_1 - z_1 t_2 }{z_2^2} \e
        \end{equation}
        and $w_2$ has a unique inverse.
        \item When $w_1$ and $w_2$ are infinitesimal, i.e. $z_1=z_2=0$. The usual constraint that $w' = \frac{w_1}{w_2} \iff w'w_2 = w_1$ leaves us with infinitely many solutions. Indeed, take $w' = z' + t' \e \in \DC$ as a candidate solution:
        \begin{equation}
        z' + t' \e = \frac{t_1 \e}{t_2 \e} \iff z' t_2 \e + t' t_2 \e^2 = t_1 \e \iff z' = t_1/t_2.
        \end{equation}
        We observe that $t'$ is unconstrained. A natural choice would be $t' = 0$, implying $\frac{t_1 \e}{t_2 \e} = \frac{t_1}{t_2}$. This is only a convenience definition, as $w_2$ has no inverse.
        \item When $w_1$ is appreciable and $w_2$ is infinitesimal. In this case we have no solution because $z' + t'\e = \frac{z_1}{t_2 \e} \iff z' t_2 \e = z_1$ which is impossible.
\end{enumerate}

Division between $w_1$ and $w_2$ is therefore defined for $w_2$ appreciable.\\

Exponentiation and rooting of \dual are given by

\begin{restatable}[\dual exponentiation]{proposition}{propexp}\label{prop:exp}
$(z + t\e)^n = z^n + n z^{n-1} t \e$
\end{restatable}
\begin{proofcontent}{propexp}
$\DC$ is a commutative ring and therefore the binomial theorem holds.
\noindent \begin{equation}
(z + t\e)^n = \sum_{k=0}^n \binom{n}{k} z^{n-k}(t\e)^k = z^n + nz^{n-1}t\e + (t\e)^2 \sum_{k=2}^n \binom{n}{k} z^{n-k}(t\e)^{k-2} = z^n + nz^{n-1}t\e
\end{equation}
\end{proofcontent}

\begin{restatable}[\dual roots]{corollary}{corroots}\label{cor:roots}
Let $z_k^{1/n}$ be the $k$-branch of the $n$-th roots of $z \neq 0$. Then the $k$-th branch of the $n$-th roots of $z + t \e$ is given by
\begin{equation}
(z + t \e)_k^{1/n} = z_k^{1/n} + \frac{t}{n(z_k^{1/n})^{n-1}} \e.
\end{equation}
\end{restatable}
\begin{proofcontent}{corroots}
Suppose $(z' + t'\e)^n = z + t \e$ for some $z', t' \in \C$. By Prop.~\ref{prop:exp}, we have
\begin{equation}
(z' + t'\e)^n = z'^n + n z'^{n-1} t' \e = z + t \e.
\end{equation}
Equating the significant and infinitesimal parts, we obtain $z'^n = z$ and $n z'^{n-1} t' = t$. Since $z \neq 0$, $z' = z_k^{1/n}$. From the second equation, we get $t' = \frac{t}{n z'^{n-1}} = \frac{t}{n(z_k^{1/n})^{n-1}}$.
\end{proofcontent}

\subsection{Automatic differentiation}\label{sec:automaticdiff}

Dual numbers are useful for automatic differentiation.  Given a complex function $f(z)$ which is holomorphic on the open set $D$, then for any $z, z_0 \in D$, we can express $f(z)$ as a sum of $f^{(n)}(z_0) := \frac{d^{n}f(z_0)}{dz^{n}}$ by means of the Taylor expansion $f(z) = \sum_{n=0}^\infty \frac{f^{(n)}(z_0) (z-z_0)^n}{n!}$. Following \cite{messelmi2015}, a \dual function $\hat{f}(z + t\e)$ is holomorphic on $D + \Z$ if and only if for any $z+t\e, z_0 \in D + \Z$,
\begin{equation}
        \hat{f}(z + t\e)= \sum_{n=0}^\infty \frac{f^{(n)}(z_0) (z + t \e -z_0)^n}{n!}
\end{equation}

\begin{restatable}[Automatic differentiation]{proposition}{propauto}\cite{messelmi2015}\label{pr:auto}
        Given a complex function $f(z)$ which is holomorphic on the open set $D$, there exists a unique \dual function $\hat{f}(z+t\e)$ which is holomorphic on $D + \Z$ and agrees with $f(z)$ for every $z \in D$. For any $z\in D$, $z+t\e \in D + \Z$, the function is given by
\begin{equation}\label{eq:auto}
        \hat{f}(z + t\e) = f(z) + t \e f^{(1)}(z).
\end{equation}
\end{restatable}
\begin{proofcontent}{propauto}
This proof is simplified and only meant for intuition; see \cite{messelmi2015} for a full formal proof.\ Given $f(z) = \sum_{n=0}^\infty \frac{f^{(n)}(z_0) (z-z_0)^n}{n!}$ and $\hat{f}(z + t\e)$ agreeing with $f$ for complex values and holomorphic on $D+\Z$, we have:
\begin{equation}
\begin{split}
        \hat{f}(z + t \e) &= \sum_{n=0}^\infty \frac{f^{(n)}(z_0) (z-z_0 + t \e)^n}{n!}
                           = \sum_{n=0}^\infty \frac{f^{(n)}(z_0) (z-z_0)^n + n (z-z_0)^{n-1} t \e}{n!}
                           \quad\textrm{(by Prop.~\ref{prop:exp})}
\end{split}
\end{equation}
this can be rewritten as
\begin{equation}
\begin{split}
        \hat{f}(z + t \e) &= \sum_{n=0}^\infty \frac{f^{(n)}(z_0) (z-z_0)^n}{n!} + \sum_{n=0}^\infty \frac{f^{(n)}(z_0) n (z-z_0)^{n-1} t \e}{n!} \\
                          &= f(z) + t \e \frac{d}{dz}(\sum_{n=0}^\infty \frac{f^{(n)}(z_0) (z-z_0)^n}{n!})
                          = f(z) + t \e f^{(1)}(z).
\end{split}
\end{equation}
\end{proofcontent}

\begin{corollary}\label{corr:elsc} \Dual holomorphic extensions of exponential, logarithm, sine and cosine are given by

\begin{multicols}{2}
\noindent \begin{equation}\label{eq:expscalar}
 \exp(z + t\e) = \exp(z)(1 + t\e)
\end{equation}

\noindent \begin{equation}\label{eq:logscalar}
 \log(z + t\e) = \log(z) + \frac t z \e
\end{equation}

\columnbreak
\noindent \begin{equation}
 \sin(z + t\e) = \sin(z) + \cos(z)t\e
\end{equation}

\noindent \begin{equation}
 \cos(z + t\e) = \cos(z) - \sin(z)t\e
\end{equation}
\end{multicols}
\end{corollary}

\section{Linear algebra over \dual numbers}\label{sec:linearalgebra}

This section provides those properties of $\D$ and $\DC$ that are needed in order to formulate postulates of quantum theory over \dual numbers and prove their consistency. We introduce an order on dual numbers that will allow us to define valid probabilities. We motivate the choice of conjugation and norm on $\DC$, and establish some properties about them---before recalling results about diagonalization of \dual Hermitian operators. Finally, we prove that \dual unitary operators are also diagonalizable, and establish their relation to \dual anti-Hermitian operators, as exponentials and logarithms of each other.

\subsection{Ordering dual and \dual numbers}\label{subsec:ordering}

In the usual definition of a `total ring order', one needs to have that (i) for all $c$, for all $a \leq b$, $a + c \leq b + c$ and (ii) for all $a \geq 0, b \geq 0$, $ab\geq 0$. We know that while $\R$ is a totally ordered ring, this is not the case for $\C$. Indeed, if $\C$ could be totally ordered we would have either $i \geq 0$ or $-i \geq 0$ and then $(\pm i)^2 = -1 \geq 0$. It turns out that there is a very natural lexicographic total ring order for $\D$:
\begin{restatable}[\dual order]{proposition}{propdualorder}\label{pr:dualorder}
There are exactly two extensions of the natural order on $\R$ that make $\D$ a totally ordered ring, determined by the choice $\e>0$ or $\e<0$.
\end{restatable}
\begin{proofcontent}{propdualorder}
Suppose $\leq_\D$ is a total ring order on $\D$ such that for any $a, b \in \R$, $a \leq b \iff a \leq_\D b$. Without loss of generality, assume $\e \geq_\D 0$ (since $\R[\e] \cong \R[-\e]$).

Suppose for the sake of a contradiction that there exists $a>0$ such that $b\e \geq_\D a$. Then by (i) $b\e-a \geq_\D 0$ and $b\e+a \geq_\D 2a>0$. By (ii) again $(b\e-a)(b\e+a)\geq_\D 0$, yielding $b^2\e^2 \geq_\D a^2 $. Since $\e^2=0$ we get $a^2\leq 0$, contradicting $a^2 > 0$. Therefore, for all $a>0$, $b\e <_\D a$.

We now show that $\leq_\D$ must be lexicographic, i.e. $a_1 + b_1 \e \leq_\D a_2 + b_2 \e$ if and only if $a_1 < a_2$  or ($a_1 = a_2$ and $b_1\leq b_2$).\\
Case $a_1 < a_2$. Then $a_2 - a_1>0$. It follows that $(b_1 - b_2)\e <_\D (a_2 - a_1)$ and hence $a_1 + b_1 \e <_\D a_2 + b_2 \e$.\\
Case $a_1 > a_2$ is symmetric, yielding $a_1 + b_1 \e >_\D a_2 + b_2 \e$.\\
Case $a_1 = a_2$ and $b_1 \leq b_2$. Then $b_2 - b_1 \geq 0$ and since $\e >_\D 0$, we have $(b_2 - b_1)\e \geq_\D 0$, as a consequence $b_1 \e \leq_\D b_2 \e$ and hence $a_1 + b_1 \e \leq_\D a_2 + b_2 \e$.\\
Case $a_1 = a_2$ and $b_1 > b_2$ is similar with $b_2 - b_1 < 0$, yielding $a_1 + b_1 \e >_\D a_2 + b_2 \e$.\\
Having verified that $\leq_\D$ is the only option, let us check that it is indeed a total ring order. \\
(i) Consider $w_1=a_1 + b_1 \e \leq_\D a_2 + b_2 \e=w_2$ and some $w_3=a_3 + b_3 \e$. By definition of $\leq_\D$ we have that
\begin{align*}
a_1 < a_2 &\Rightarrow (a_1 + a_3) < (a_2 + a_3) \Rightarrow  w_1+w_3<_\D  w_2+w_3.\\
a_1 = a_2\textrm{ and }b_1 \leq b_2 &\Rightarrow (a_1 + a_3) = (a_2 + a_3)\textrm{ and }(b_1 + b_3) \leq (b_2 + b_3) \Rightarrow w_1+w_3 \leq_\D  w_2+w_3
\end{align*}
(ii) Consider $w_1=a_1 + b_1 \e \geq_\D 0$ and $w_2= a_2 + b_2 \e \geq_\D 0$, since the order is lexicographic, $a_1 \geq 0$ and $a_2 \geq 0$, and so $a_1 a_2 \geq 0$. Now, notice that
\begin{equation}
w_1w_2 \geq_\D 0 \Leftrightarrow a_1 a_2 +(a_1 b_2 + a_2 b_1)\e \geq_\D 0.
\end{equation}
We distinguish three cases. If $a_1 a_2 > 0$, the significant part is positive, so $w_1w_2 \geq_\D 0$. If $a_1=0$ and $a_2\geq 0$ (or symmetrically $a_2=0$ and $a_1\geq 0$), then $w_1=b_1\e$ with $b_1\geq 0$, and $w_1w_2 = a_2 b_1\,\e \geq_\D 0$.
\end{proofcontent}

\subsection{Linear algebra}\label{subsec:algebra}

In \cite{messelmi2015} the author suggests four different possible ways of conjugating dual complex numbers:
\begin{multicols}{4}
\noindent \begin{equation}
w^* = z^* + t^*\e
\end{equation}

\columnbreak

\noindent \begin{equation}
\til{w} = z - t\e
\end{equation}

\columnbreak

\noindent \begin{equation}
\bar{\til{w}} = z^* - t^*\e
\end{equation}

\columnbreak

\noindent \begin{equation}
\bar{w} = z^*(1-t\e/z)
\end{equation}
\end{multicols}

The $\bar{w}$ option is the only one which satisfies both $\bar{w}w \in \R$ and $\Re(\Sig(w)) = \Re(\Sig(\bar{w}))$, but it is non-linear. Since our purpose is to extend quantum mechanics, which is fundamentally linear, we must give up this option and favour the first property over the second. Indeed, the two options $\til{w}$ and $\bar{\til{w}}$ lead to norms having an imaginary part, which is difficult to interpret in quantum mechanics. We are left with the $w^*$ option, which leads to norms having an infinitesimal part. Those can be understood as representing infinitesimal variations of the norm, we opt for this definition as done in \cite{qi2023}.
\begin{definition}
Let $n \in \N$. We define $\DC^n$ to be the module of $n$-dimensional vectors with numbers in $\DC$, equipped with the inner-product:
\begin{equation}
 \braket{\psi|\phi} = \sum_k \psi_k^*\phi_k
\end{equation}

\end{definition}

The norm $|w|$ of $w = z + t\e$ is defined as $|w| = \sqrt{w^* w}$. We can observe the following properties

\begin{multicols}{3}
\noindent \begin{equation}
(w_1 + w_2)^* = w_1^* + w_2^*
\end{equation}

\noindent \begin{equation}
w + w^* = 2 \Re(z) + 2 \Re(t) \e
\end{equation}

\noindent \begin{equation}\label{eq:eq0}
|w| = 0 \iff z = 0
\end{equation}

\columnbreak

\noindent \begin{equation}
(w_1  w_2)^* = w_1^* w_2^*
\end{equation}

\noindent \begin{equation}\label{eq:sqnorm}
w^*w = |z|^2 + 2 \Re(z^* t) \e \in \D
\end{equation}

\noindent \begin{equation}\label{eq:geq0}
|w| \geq 0
\end{equation}

\columnbreak

\noindent \begin{equation}
(w^*)^* = w
\end{equation}

\noindent \begin{equation}\label{eq:modul}
|w| = \sqrt{w^* w} = |z| + \frac{\Re(z^* t)}{|z|} \e
\end{equation}

\end{multicols}

In particular equations \eqref{eq:geq0} and \eqref{eq:modul} entail the following propositions.

\begin{proposition}[\Dual norm]
 For every $\ket{\psi_\e} = \ket{\psi} + \e \ket{\phi} \in \DC^n$, we have

 $$\braket{\psi_\e|\psi_\e} = \braket{\psi|\psi} + 2\Re(\braket{\psi|\phi}) \e$$

 and therefore
\begin{equation}\label{eq:norm}
 ||\ket{\psi_\e}|| = \sqrt{\braket{\psi_\e|\psi_\e}} = ||\ket\psi|| + \frac{\Re(\braket{\psi|\phi})}{||\ket\psi||}\e.
\end{equation}
We say that $\ket{\psi_\e}$ is {\em unit} norm iff $||\ket{\psi_\e}||=1$.
\end{proposition}

\begin{corollary}[\Dual norm non-negativity]
 For any $\ket{\psi} \in \DC^n$, $||\ket\psi||^2 = \sum_{k=1}^n \psi_k^* \psi_k \geq_\D 0$.
\end{corollary}
Notice that the norm of a vector is infinitesimal if and only if all entries are infinitesimal. In this case the norm of the vector is zero, and the vector itself is said to be infinitesimal. Conversely, the norm of a vector is appreciable if and only if at least one entry is appreciable. In this case the vector itself is said to be appreciable. Two vectors are called orthogonal when their inner product is infinitesimal.\\

We define the adjoint as the transpose conjugate, i.e. $A^\dagger = (A^T)^*$. As usual an operator $H$ is said to be Hermitian if $H^\dagger = H$ and anti-Hermitian if $H^\dagger = -H$, for both its significant part and infinitesimal part. An operator $U$ is unitary when $U^\dagger U = I$. It follows that:\\

\begin{restatable}[\Dual unitarity]{proposition}{propunitary}\label{pr:unitary}
The following propositions are equivalent for $U_\e$ a square linear operator:
\begin{enumerate}
 \item[(a)] $U_\e$ is \dual unitary, i.e.\ $U_\e^\dagger U_\e = I$.
 \item[(b)] $U_\e$ preserves the inner product, i.e.\ $(U_\e \ket\psi)^\dagger (U_\e \ket\phi) = \braket{\psi|\phi}$ and therefore the norm $||U_\e \ket\psi|| = ||\ket\psi||$.
 \item[(c)] Rows of $U_\e$ form an orthonormal basis.
 \item[(d)] Columns of $U_\e$ form an orthonormal basis.
 \item[(e)] $U_\e = (I + i \e H)U$ where $U$ is a complex unitary and $H$ is a complex Hermitian.
\end{enumerate}
\end{restatable}
\begin{proofcontent}{propunitary}
 For point equivalence between point (e) and (a), suppose $U_\e = U + i \e HU$.
\begin{equation}
 \begin{split}
  U_\e^\dagger U_\e = I &\iff U^\dagger U + \e U^\dagger (i \e HU) + \e (i \e HU)^\dagger U = I\\
                        &\iff \begin{cases}
      U^\dagger U = I \\
      i\e U^\dagger HU - i\e U^\dagger H^\dagger U = 0
   \end{cases}\\
                        &\iff \begin{cases}
      U \text{ is a unitary.} \\
      H \text{ is a Hermitian.}
   \end{cases}\\
 \end{split}
\end{equation}
Points from (a) to (e) are easily seen to be equivalent by following a reasoning analogous to the case of complex numbers.
\end{proofcontent}

In the rest of this section, we prove that the exponential of \dual anti-Hermitian is unitary and that the logarithm of a \dual unitary is anti-Hermitian. We first introduce the exponential of a \dual operator and recall a known result about \dual Hermitian diagonalization. Next, we give two new results about \dual unitary operators.

\begin{restatable}[\Dual matrix exponentiation]{proposition}{propmatexp}\label{pr:matexp}
The \dual matrix exponential always exists.
\end{restatable}
\begin{proofcontent}{propmatexp}
For scalar exponentiation in Eq.~\eqref{eq:expscalar}, we used the \dual extension from Eq.~\eqref{pr:auto} which works thanks to binomial theorem and, therefore, commutativity. Commutativity is notoriously not a property of linear operators in general. Without commutativity, we get
\begin{equation}
\exp(A + B \e) = \sum_{n=0}^\infty \frac 1 {n!} (A + B\e)^n = \sum_{n=0}^\infty \frac{A^n}{n!} + \e \sum_{n=0}^\infty \frac 1 {n!} \sum_{k=0}^{n-1} A^k B A^{n-1-k}
\end{equation}
The significant part is the exponential of the complex matrix $A$ which is known to always converge. Using the usual operator norm, we see that the infinitesimal part also converges as
\begin{equation}
\begin{split}
||\sum_{n=0}^\infty \frac 1 {n!} \sum_{k=0}^{n-1} A^k B A^{n-1-k}|| &\leq \sum_{n=0}^\infty \frac 1 {n!} \sum_{k=0}^{n-1} ||A^k||\, ||B||\, ||A^{n-1-k}||\\
&\leq ||B||\sum_{n=0}^\infty \frac{n ||A||^{n-1}}{n!}
= ||B|| \exp(||A||)
\end{split}
\end{equation}
\end{proofcontent}

It is interesting to note that a proof identical to that of \cite{liu2025} for the case of dual numbers can be used to show that the exponential of \dual matrix $A + B\e$ is $\exp(A) + \e L_{\exp}(A, B)$ where $L_{\exp}(A, B)$ is the Fréchet derivative of the exponential function of $A$ in the direction of $B$.\\

\begin{proposition}[\cite{qi2021}]\label{pr:qi21}
 Suppose $H$ is a Hermitian operator. Then there exists an orthonormal \dual basis  $ \{\ket{j_\e} \}$  and $n$ dual eigenvalues $\lambda_j + \e \mu_j$ of $H$ such that $H = \sum_j (\lambda_j + \e \mu_j) \ketbra{j_\e}{j_\e}$.
\end{proposition}

We now establish that \dual unitary operators, just like \dual Hermitian operators, are diagonalizable and deduce that their logarithms are anti-Hermitian.

\begin{restatable}[Spectral theorem for \dual unitary]{proposition}{propspecunit}\label{th:specunit}
 A unitary operator is unitarily diagonalizable, i.e.\ for every \dual unitary operator $U_\e$, there is a \dual orthonormal eigenbasis $ \{\ket{j_{U_\e}} \}$  with associated \dual eigenvalues $e^{i\theta_j}(1 + \e i \mu_j)$ where $\theta_j$ and $\mu_j$ are real numbers, and such that $U_\e = \sum_j e^{i\theta_j}(1 + \e i \mu_j) \ketbra{j_{U_\e}}{j_{U_\e}}$.
\end{restatable}

\begin{proofcontent}{propspecunit}
Let $U_\e = U + i\e J U$ be a \dual unitary where $U = \exp(iH)$ is a complex unitary operator and $H$, $J$ are Hermitian operators. We are going to build an eigenbasis $\ket{j_{U_\e}}$ for $U_\e$ from an eigenbasis of $H_{\e} = H + \e J$.

$[H_\e]$. Indeed, by Prop.~\ref{pr:qi21}, we have that $H_{\e}$ is diagonalizable. Let $\ket{j_\e} = \ket{j_0} + \e \ket{j_1}$ be an orthonormal eigenbasis with the eigenvalues $\theta_j + \e \mu_j$ where $\theta_j, \mu_j \in \R$, in which $ H_{\e}$ is diagonal. Then, we have the following eigenvalue equation,

\begin{equation}\label{eq:Heigvalue}
   \begin{split}
        (H + \e J) (\ket{j_0} + \e \ket{j_1}) = (\theta_j + \e \mu_j)(\ket{j_0} + \e \ket{j_1})\\
        \iff H \ket{j_0} + \e( J \ket{j_0} + H \ket{j_1} ) = \theta_j \ket{j_0} + \e( \theta_j \ket{j_1} + \mu_j \ket{j_0})
   \end{split}
\end{equation}

From the significant part of Eq.~\eqref{eq:Heigvalue}, we find that $\ket{j_0}$ forms an eigenbasis for $H$. Then we can uniquely express $J$ in the $ \{\ket{j_0} \}$ eigenbasis, i.e. $J = \sum_{jk} h_{kj} \ketbra{k_0}{j_0}$ where $h_{kj}^* = h_{jk}$. For the infinitesimal part of Eq.~\eqref{eq:Heigvalue}, we see that

\begin{equation}
 H \ket{j_1} + J \ket{j_0} = \theta_j \ket{j_1} + \mu_j \ket{j_0} \iff (H - \theta_j I)\ket{j_1} = (\mu_j I - J) \ket{j_0}
\end{equation}

A prior remark is that this implies that for every $j \neq k$,
\begin{equation}
 \bra{k_0}(H - \theta_j I)\ket{j_1} = \bra{k_0}(\mu_j I - J) \ket{j_0} \iff  (\theta_k - \theta_j)\braket{k_0|j_1} = -h_{kj},
\end{equation}
which implies that for each $j \neq k$ such that $\theta_j = \theta_k$, $h_{jk} = 0$.

$[U_\e]$. With this prior remark in mind, we go back to constructing an orthonormal eigenbasis for $U_\e$. First, $U = \exp{(i H)}$ can again be diagonalized in the same orthonormal basis $\{\ket{j_0} \}$, with corresponding eigenvalues $ \lambda_{j} = e^{i \theta_{j}}$. Hence, what is left is to find is $\{\ket{\tilde{j_1}} \}$ and $ \{\sigma_{j}\}$ satisfying the eigenvalue equation:
\begin{equation}\label{eq:eigeneq}
    U_{\e} ( \ket{j_0} + \e \ket{\tilde{j_1}}) =  ( \lambda_{j} + \e \sigma_{j}) ( \ket{j_0} + \e \ket{\tilde{j_1}}).
\end{equation}
and such that the $\ket{j_{U_\e}} = \ket{j_0} + \e \ket{\tilde{j_1}}$ are orthonormal.\\
Second we analyze the infinitesimal part of the equation above. We obtain
\begin{equation}
 (U - \lambda_j I)  \ket{\tilde{j_1}} = (\sigma_j I - iJU) \ket{j_0}\label{eq:infinitesimaleigeneq}
\end{equation}
Since $U = \sum_{k} \lambda_k \ketbra{k_0}{k_0}$, the LHS must necessarily be orthogonal to $\ket{j_0}$ which means that
\begin{equation}
 \bra{j_0} (\sigma_j I - iJU) \ket{j_0} = 0
 \iff \sigma_j = i \lambda_j h_{jj}
\end{equation}
Expressed in terms of $\alpha_{jk} := \braket{k_0|\tilde{j_1}}$, Eq.~\eqref{eq:infinitesimaleigeneq} can be restated as
\begin{equation}\label{eq:alphj3}
\begin{split}
 \bra{k_0} (U - \lambda_j I)  \ket{\tilde{j_1}} = \bra{k_0} (\sigma_j I - iJU) \ket{j_0}
 \iff& (\lambda_k - \lambda_j) \alpha_{jk} = \delta_{kj} \sigma_j - i \lambda_j h_{kj}
\end{split}
\end{equation}

In the case $\lambda_j \neq \lambda_k$, this means $\alpha_{jk} = \frac{i \lambda_j h_{kj}}{\lambda_j - \lambda_k}$. In the case $j \neq k$ but $\lambda_j = \lambda_k$, our prior remark showed that $h_{jk} = 0$, hence any $\alpha_{jk} \in \C$ is a solution. The same goes for every $\alpha_{jj}$. Choosing $\alpha_{jk} = 0$ when $\lambda_j = \lambda_k$ therefore defines a valid eigenbasis $\ket{j_{U_\e}} = \ket{j_0} + \e \ket{\tilde{j_1}}$ for $U_\e$ with eigenvalues $e^{i\theta_j}(1+i\e h_{jj})$.

$[\textrm{Orthonormality}]$. We are left to show that $\ket{j_{U_\e}}$ forms an orthonormal basis. With our definition, when $\lambda_j = \lambda_k$, $\alpha_{jk}^* = 0 = -\alpha_{kj}$. When $\lambda_j \neq \lambda_k$,
\begin{gather*}
\alpha_{jk}^* = \frac{-i \lambda_j^* h_{kj}^*}{\lambda_j^* - \lambda_k^*} = \frac{-i \lambda_j^{-1} h_{jk}}{\lambda_j^{-1} - \lambda_k^{-1}} = \frac{-i \lambda_k h_{jk}}{\lambda_k - \lambda_j} = -\alpha_{kj}\\
\implies \braket{j_{U_\e}|k_{U_\e}} = \braket{j_0|k_0} + \e \braket{\tilde{j_1}|k_0} + \e \braket{j_0|\tilde{k_1}} = \delta_{jk} + \e (\alpha_{kj}^* + \alpha_{jk}) = \delta_{jk}
\end{gather*}

$[\textrm{Diagonalization recap.}]$. By construction we have $\theta_j, h_{jj} \in \R$ and
\begin{equation}
\begin{split}
\sum_j \lambda_j (\ketbra{j_0}{\tilde{j_1}} + \ketbra{\tilde{j_1}}{j_0})
&= \sum_{jk} \lambda_j (\alpha_{jk}^* \ketbra{j_0}{k_0} + \alpha_{jk} \ketbra{k_0}{j_0}) \\
&= - \sum_{\substack{jk\\ \lambda_j \neq \lambda_k}} i h_{jk}^* \frac{\lambda_k \lambda_j}{\lambda_k - \lambda_j} \ketbra{j_0}{k_0} + \sum_{\substack{jk\\ \lambda_j \neq \lambda_k}} i h_{kj} \frac{\lambda_j^2}{\lambda_j - \lambda_k} \ketbra{k_0}{j_0} \\
&= - \sum_{\substack{jk\\ \lambda_j \neq \lambda_k}} i h_{kj} \frac{\lambda_j \lambda_k}{\lambda_j - \lambda_k} \ketbra{k_0}{j_0} + \sum_{\substack{jk\\ \lambda_j \neq \lambda_k}} i h_{kj} \frac{\lambda_j^2}{\lambda_j - \lambda_k} \ketbra{k_0}{j_0} \\
&= \sum_{\substack{jk\\ \lambda_j \neq \lambda_k}} i h_{kj} \lambda_j \frac{\lambda_j - \lambda_k}{\lambda_j - \lambda_k} \ketbra{k_0}{j_0} = \sum_{\substack{jk\\ \lambda_j \neq \lambda_k}} i h_{kj} \lambda_j \ketbra{k_0}{j_0} \\
&= \sum_{\substack{jk\\ j \neq k}} i h_{kj} \lambda_j \ketbra{k_0}{j_0} \quad\textrm{by the prior remark.}
\end{split}
\end{equation}
This gives us
\begin{equation}
\begin{split}
\sum_j e^{i\theta_j}(1 + \e i h_{jj}) \ketbra{j_{U_\e}}{j_{U_\e}}
&= \sum_j \lambda_j \ketbra{j_0}{j_0}
+ \e i \lambda_j h_{jj} \ketbra{j_0}{j_0}
+ \e \lambda_j (\ketbra{j_0}{\tilde{j_1}} + \ketbra{\tilde{j_1}}{j_0}) \\
&= U + \e i \sum_j h_{jj} \lambda_j \ketbra{j_0}{j_0} + \e \sum_{\substack{jk\\ j \neq k}} i h_{kj} \lambda_j \ketbra{k_0}{j_0} \\
&= U + i \e \sum_{jk} h_{kj} \lambda_j \ketbra{k_0}{j_0} \\
&= U + i \e (\sum_{jk} h_{kj} \ketbra{k_0}{j_0})(\sum_j \lambda_j \ketbra{j_0}{j_0})\\
&= U + i \e J U
\end{split}
\end{equation}
\end{proofcontent}

Thanks to the diagonalization of unitary operators, we can establish the following result.

\begin{restatable}[\Dual unitary as exponential of anti-Hermitian]{proposition}{prophermunit}\label{pr:hermunit}
The exponential of anti-Hermi-tian linear operator is a unitary. Moreover, every unitary $U_\e$ is of the form $\exp(iH_{\e})$ where $H_{\e}$ is dual-complex Hermitian. $H_{\e}$ is unique modulo $2\pi$.
\end{restatable}

\begin{proofcontent}{prophermunit}
To prove the first statement, let $H_{\e}$ be a Hermitian operator. Then we have

\begin{equation}
 \exp(iH_{\e})^\dagger = (\sum_{n=0}^\infty \frac{(iH_{\e})^n}{n!})^\dagger = \sum_{n=0}^\infty \frac{(-iH_{\e}^\dagger)^n}{n!} = \exp(-iH_{\e})
\end{equation}

This shows us that the adjoint of the exponential of $iH$ is its inverse:

\begin{equation}
 \exp(iH_{\e})^\dagger \exp(iH_{\e}) = \exp(-iH_{\e}) \exp(iH_{\e}) = \exp(-iH_{\e}+iH_{\e}) = \exp(0) = I
\end{equation}

For the second statement, let $U_\e = \sum_j (\lambda_j + \e \sigma_j) \ketbra{j_{U_\e}}{j_{U_\e}}$ be a \dual unitary operator diagonalized as in Prop.~\ref{th:specunit}. We can write $\lambda_j + \e \sigma_j = e^{i\theta_j}(1 + i\e\mu_j)$ where $\theta_j, \mu_j \in \R$. Therefore, $\log(U_\e) = \sum_j \log(\lambda_j + \e \sigma_j)\ketbra{j_{U_\e}}{j_{U_\e}} = \sum_j (i\theta_j + \e i\mu_j)\ketbra{j_{U_\e}}{j_{U_\e}}$ (by Eq.~\eqref{eq:logscalar}) which is indeed anti-Hermitian.

Uniqueness modulo $2\pi$ can be seen from the fact that the other solutions $z+t\e$ of $\log(e^{i\theta_j}(1 + i\e\mu_j))$, i.e. $z,t \in \C$ such that $\exp(z + t\e) = e^{i\theta_j}(1 + i\e\mu_j)$ are those such that $e^{i\theta_j} = e^{z}$ and $t = i \mu_j$ by using Eq.~\eqref{eq:expscalar}.

\end{proofcontent}

Another result is required to extend the measurement postulate. We say the operator $E$ is appreciably semipositive if $\braket{\psi|E|\psi}$ is positive and appreciable or exactly zero for every vector $\ket\psi$.

\begin{restatable}[\Dual semipositivity]{proposition}{propsemipos}\label{pr:semipos}
 For every \dual linear operator $M_\e$, $M_\e^\dagger M_\e$ is appreciably semipositive.
\end{restatable}
\begin{proofcontent}{propsemipos}
Let $\ket\psi$ be an arbitrary vector and $\ket\phi = \sum_j \phi_j \ket{j} = M_\e \ket\psi$ for an orthonormal basis $\ket{j}$. Then $\braket{\psi|M_\e^\dagger M_\e|\psi} = \braket{\phi|\phi} = \sum_j \phi_j^*\phi_j$. From Eqs \eqref{eq:eq0}, \eqref{eq:sqnorm} and \eqref{eq:geq0} combined, we have that for all $j$, $\phi_j^* \phi_j$ is either appreciable or exactly zero and always greater or equal to zero. It follows that $M_\e^\dagger M_\e$ is appreciably semipositive.
\end{proofcontent}

The Stinespring dilation theorem also continues to hold with \dual numbers:
\begin{proposition}[\Dual Stinespring dilation]\label{pr:dualstinespring}
For every family of \dual operators $\{M_m^\e\}$, let $V_\e:=\sum_m \ket{m}\bra{0}\otimes M_m^\e$. We have that
\begin{align*}
        M_m^\e \ket{\psi}&=(\bra{m}\otimes I)V_\e(\ket{0}\otimes\ket{\psi})\\
        \sum_m {M_m^\e}^\dagger M_m^\e=I &\Leftrightarrow V_\e^\dagger V_\e=I
\end{align*}
If the completeness relation is satisfied, $V_\e$ is a \dual isometry and can be extended to a \dual unitary operator $U_\e$, called a {\em Stinespring dilation}.
\end{proposition}
\begin{proof}
The proof is the same as for the case of complex numbers.
\end{proof}

\section{Postulates of quantum theory over \dual numbers}\label{sec:postulates}

In this section we take the postulates of conventional quantum theory, and replace the complex scalar field, by the \dual ring. This defines a quantum theory over \dual numbers, which we refer to as \dual quantum theory for short.

\begin{postulate}[\Dual states]\label{post:dualstates}
At all times, a \dual quantum system can be described by a \dual state $\ket{\psi_\e}$, defined to be a unit \dual vector.
\end{postulate}

Every \dual state $\ket{\psi_\e}$ decomposes into a significant part $\ket{\psi}$ and an infinitesimal part $\ket{\phi}$, i.e. $\ket{\psi_\e}=\ket{\psi}+\e\ket{\phi}$. By Eq.~\eqref{eq:norm}:
\begin{align*}
||\ket\psi + \e \ket\phi|| = 1 &\Leftrightarrow ||\ket\psi|| + \e \Re(\braket{\psi|\phi}) = 1\\
& \Leftrightarrow (i)\quad ||\ket\psi||=1\quad  \textrm{ and }\quad(ii)\quad \Re(\braket{\psi|\phi})= 0
\end{align*}
Interestingly, if we choose to interpret the infinitesimal part as some kind of perturbation of the significant part, i.e. $\ket{\phi}=A\ket{\psi}$, then this perturbation $A$ is necessarily anti-Hermitian. Indeed, by (ii):
\begin{align*}
\Re(\braket{\psi|\phi}) = \Re(\braket{\psi|A|\psi}) = 0 \Leftrightarrow \braket{\psi|(iA)|\psi} \in \R
\Leftrightarrow A \textrm{ is anti-Hermitian.}
\end{align*}

\begin{postulate}[\Dual evolutions]\label{post:dualevolutions}
Over any finite period of time, the evolution of a \dual quantum system is described by the application of a \dual evolution $U_\e$, defined to be a unitary \dual operator.
\end{postulate}

Every \dual evolution $U_\e$ decomposes into a significant part $U$ and an infinitesimal part $V$, i.e. $U_\e=U+\e V$. Recall from our Prop.~\ref{pr:unitary}:

\begin{align*}
        U_\e^\dagger U_\e = I &\Leftrightarrow (i) \quad U^\dagger U = I\quad  \textrm{ and }\quad (ii)\quad U^\dagger V \textrm{ is anti-Hermitian}\\
        &\Leftrightarrow (i) \quad U^\dagger U = I\quad  \textrm{ and }\quad (ii)\quad V=iHU \text{ with $H$ Hermitian}
\end{align*}

This fits the interpretation that the infinitesimal part of a \dual state is indeed produced by an anti-Hermitian perturbation of the significant part of the unitary evolution. As a first example consider:

\begin{align*}
        U_\e = \begin{pmatrix}
        -im\e & 1 \\
        1 & -im\e \\
        \end{pmatrix}\quad
        U= \begin{pmatrix} 0 & 1 \\ 1 & 0 \end{pmatrix} \quad V= \begin{pmatrix} -im & 0 \\ 0 & -im \end{pmatrix} \quad H= \begin{pmatrix} 0 & -m \\ -m & 0 \end{pmatrix}
\end{align*}
Notice that $U_\e$ is perfectly unitary in the sense of Prop.~\ref{pr:unitary}, it feels dangerously remote from a unitary in the conventional sense, i.e. if we replace $\e$ by $h\in\mathbb{R}$ the obtained $U_h$ is not a unitary. Still we will prove in App.~\ref{sec:consistency}, Prop.~\ref{pr:unitaryextcorr} that $U_h$ is at most $O(h^2)$ away from a unitary, here
\begin{align*}
\til{U}_h = \begin{pmatrix}
        -i\sin(mh) & \cos(mh) \\
        \cos(mh) & -i\sin(mh) \\
        \end{pmatrix}.
\end{align*}
Another way to put this is that $U_\e$ could equally well have been presented as $\til{U}_\e$, because over \dual numbers we just have $U_\e = \til{U}_\e$. The apparent relaxation of unitarity is therefore innocuous.\\
In conventional quantum mechanics, the evolution postulate is equivalently expressed in continuous time through the Schrödinger equation. From our Prop.~\ref{pr:hermunit}, the same equivalence holds here.\\
{\begin{postulate}{(continuous-time evolutions)}\label{post:continuous}
The continuous time evolution of a closed system is described by the Schrödinger equation,
\begin{equation}
 i \hslash \dstate{\psi_\e}{t} = H_\e\ket{\psi_\e}
\end{equation}
where $H_\e$ is a \dual Hermitian linear operator.
\end{postulate}}

\begin{postulate}[\Dual measurements]\label{post:dualmeasurements}
A family of \dual operators $\{M_m^\e\}$ describes a \dual measurement iff it is bound by the completeness relation $\sum_m {M_m^\e}^\dagger M_m^\e = I$. The outcome $m$ appears with the dual probability $p(m) = \bra\psi {M_m^\e}^\dagger M_m^\e \ket\psi$. The resulting \dual state is
\begin{equation}\label{def:meas}
 \ket{\psi'_\e} = \frac{M_m^\e \ket\psi}{\sqrt{p(m)}}.
\end{equation}
\end{postulate}

Every \dual measurement operator $M_m^\e$ decomposes into a significant part $M_m$ and an infinitesimal part $N_m$, i.e.\ $M_m^\e = M_m + \e N_m$. Expanding the completeness relation:
\begin{align*}
\sum_m (M_m + \e N_m)^\dagger (M_m + \e N_m) = I
&\Leftrightarrow \sum_m M_m^\dagger M_m + \e \sum_m (M_m^\dagger N_m + N_m^\dagger M_m)= I\\
&\Leftrightarrow (i)\quad \sum_m M_m^\dagger M_m = I \quad \textrm{ and }\quad (ii)\quad  M_m^\dagger N_m \textrm{ anti-Hermitian.}
\end{align*}
To ensure our \dual quantum theory does not cause unwanted behaviors, we must check that probabilities are positive and sum to one. We must also check that the resulting \dual state is well-defined, i.e. that it does not get renormalized by an infinitesimal probability. Fortunately, none of these bad behaviors arise, as we will show in App.~\ref{sec:consistency}, Prop.~\ref{pr:consistency}.

As in conventional quantum theory, we can see that when we have exactly one measurement operator, then it must be unitary, which means that Postulate~\ref{post:dualevolutions} can be deduced from Postulate~\ref{post:dualmeasurements}. {On the other hand, we will show in App.~\ref{sec:consistency} that, through Stinespring dilation, dual-complex measurement operators can also be seen as arising from a dual-complex unitary evolution of a larger closed system. Moreover, Postulate~\ref{post:continuous} implies that the infinitesimal part of the dual-complex unitary evolution arises from the infinitesimal part of the energy observable, which, as a consequence of Postulate \ref{post:dualmeasurements}, can have dual-complex eigenvalues and eigenvectors. In App.~\ref{sec:consistency} it becomes apparent that $\e$ can be interpreted as an infinitesimal variation of a parameter of the experiment being carried out.}

\begin{postulate}[Composite \dual quantum systems]
Given two \dual states $\ket{\psi_\e}^A$ and $\ket{\psi_\e}^B$ describing two \dual quantum systems $A$ and $B$ that have not yet interacted, the composite \dual quantum system $AB$ is described by the \dual state $\ket{\psi_\e}^A\otimes\ket{\psi_\e}^B$ with $\otimes$ the usual tensor aka Kronecker product.
\end{postulate}
Notice that this postulate preserves the unit norm. Indeed, if we have two \dual states $\ket{\psi_\e}^A$ and $\ket{\psi_\e}^B$, the norm of $\ket{\psi_\e}^A\otimes \ket{\psi_\e}^B$ is still $\sqrt{\braket{\psi_\e|\psi_\e} \braket{\phi_\e|\phi_\e}} = 1$.

\section{Consistency of \dual quantum theory}\label{sec:consistency}

\subsection*{Self-consistency}

We first prove that \dual quantum theory is internally consistent, in the sense that \dual evolutions and measurements of a system preserves the fact of being a \dual state, and that measurement outcome probabilities remain well-behaved.

\begin{proposition}\label{pr:consistency}
A system that starts in a \dual state and is evolved according to \dual evolutions and \dual measurements, remains described by a \dual state at all times.
Moreover, probabilities from any \dual measurement are either zero or appreciably positive, and sum to one.
\end{proposition}
\begin{proof}
By Prop.~\ref{pr:unitary}, a \dual unitary preserves the norm. For measurements, from Eq.~\eqref{def:meas} norm preservation is given by
\begin{equation}
\norm{\frac{M_m^\e \ket\psi}{\sqrt{p_\e(m)}}} = \norm{\frac{M_m^\e \ket\psi}{\sqrt{\braket{\psi|{M_m^\e}^\dagger M_m^\e|\psi}}}} = \frac{||M_m^\e \ket\psi||}{||M_m^\e\ket\psi||} = 1.
\end{equation}
where $p_\e(m)$ is the probability of outcome $m$.

Probabilities sum to one since $\sum_m {p_\e(m)} = \sum_m \braket{\psi|{M_m^\e}^\dagger M_m^\e|\psi} = \braket{\psi|(\sum_m {M_m^\e}^\dagger M_m^\e)|\psi} = \braket{\psi|I|\psi} = 1$ and are either zero or appreciably positive by Prop.~\ref{pr:semipos}.
\end{proof}

If we decompose into significant and infinitesimal parts as $p_\e(m) = p_m + \e \dot{p}_m$, $M_m^\e = M_m + \e \dot{M}_m$ and $\ket{\phi_\e} = \ket\phi + \e \ket{\dot\phi} = M_m^\e \ket\psi$, we can easily check that $\ket\phi = M_m \Sig(\ket\psi)$, which only depends on the significant part of $M_\e$ and $\ket\psi$. Interestingly, if we think of the infinitesimal parts as representing the derivatives with respect to a variable $z$ of the significant parts in each of these decompositions, then by Eq.~\ref{eq:inverse} and the quotient rule,
\begin{equation}
\frac{M_m^\e \ket\psi}{\sqrt{p_\e(m)}} = \frac{\ket\phi}{\sqrt{p_m}} + \e \frac{\ket{\dot\phi}p_m - \tfrac{1}{2}\ket\phi \dot{p}_m}{p_m^{3/2}} = \frac{\ket\phi}{\sqrt{p_m}} + \e \frac{\mathrm d}{\mathrm d z}\left(\frac{\ket\phi}{\sqrt{p_m}}\right)
\end{equation}

\subsection*{Consistency with conventional quantum theory}

In Sec.~\ref{sec:automaticdiff}, we showed how holomorphic functions can be extended from complex numbers to \dual. In a similar vein, we now show how conventional quantum theory operators parametrized by $h$, can be extended to become \dual quantum theory operators.

We subsequently answer concerns that the \dual quantum theory operators obtained that way, or otherwise, may only be valid in \dual quantum theory. E.g. \dual unitary operators and \dual complete measurements are not in general unitary and complete in the conventional sense---when $\e$ is replaced by a real number $h$. We give a general method to add an $O(h^2)$ correction which ``unitarizes/completes'' them, in the conventional sense.

\begin{definition}[\Dual extension and complex correction for unitaries]\label{def:unitaryextcorr}
Let $U_h$ be a conventional unitary which is analytic in its real parameter $h$. Let $U:=U_0$ and recall that $\frac{\mathrm{d}U_h}{\mathrm{d}h}\big|_{h=0} = i H U$ for some Hermitian $H$. The {\em \dual extension} of $U_r$ is defined by:
\begin{align*}
\hat{U}_{\e}:=U+\e i H U.
\end{align*}
Let $U_\e$ be a \dual unitary. Recall that $U_\e=U+\e i H U$.  Its {\em complex correction} is the conventional unitary:
\begin{equation}
\til{U}_h := U_h + h^2\sum_{n=0}^\infty \frac{1}{(n+2)!} (i H)^{n+2}h^n U = \exp(i h H) U
\end{equation}
\end{definition}

\begin{proposition}[\Dual extension and complex correction for unitaries]\label{pr:unitaryextcorr}
Consider $U_h$ a conventional unitary analytic in its real parameter $h$. We have $U_h =\hat{U}_{h}+O(h^2)$.\\
Consider $U_\e$ a \dual unitary. We have $U_\e\big|_{\e=h} + O(h^2)=\til{U}_{h}$.
\end{proposition}
\begin{proof}
First part, the Taylor expansion of $U_h$ gives
\begin{equation}
U_h = U + h\,iHU + O(h^2) = \hat{U}_h + O(h^2).
\end{equation}
Second part,
\begin{align*}
\til{U}_h &= \exp(ihH)U= \left(I + ihH + O(h^2)\right)U \quad\textrm{by def. of matrix exponentials}\\
&= U + ihHU + O(h^2)= U_\e\big|_{\e=h} + O(h^2).
\end{align*}
\end{proof}

\begin{definition}[\Dual extension and complex correction for measurements]\label{def:measextcorr}
Let $\{M_m^h\}$ be a conventional complete measurement analytic in its real parameter $h$. Let $M_m:=M_m^0$ and $N_m:=\frac{\mathrm{d}M_m^h}{\mathrm{d}h}\big|_{h=0}$. The {\em \dual extension} of $\{M_m^h\}$ is $\{\hat{M}_m^{\e}\}$ with
\begin{align*}
\hat{M}_m^{\e}:=M_m+\e N_m.
\end{align*}
Let $\{M_m^\e\}$ be a \dual complete measurement with Stinespring dilation $U_\e = U + \e iHU$ (Prop.~\ref{pr:dualstinespring}). The corresponding {\em complex correction} is the conventional measurement $\{\til{M}_m^h\}$ with:
\begin{equation}
\til{M}_m^h := (\bra{m} \otimes I)\, \til{U}_h\, (\ket{0} \otimes I).
\end{equation}
\end{definition}
\begin{proof}
$\{\hat{M}_m^{\e}\}$ is complete as \dual measurement since
\begin{align*}
\sum_m (\hat{M}_m^\e)^\dagger \hat{M}_m^\e
    &= \sum_m (M_m^\dagger + \e N_m^\dagger)(M_m + \e N_m)
    = \sum_m M_m^\dagger M_m + \e\sum_m (M_m^\dagger N_m + N_m^\dagger M_m) \\
    &= \sum_m M_m^\dagger M_m + \e\,\frac{\mathrm{d}}{\mathrm{d}h}\bigg|_{h=0}\sum_m (M_m^h)^\dagger M_m^h = I.
\end{align*}
$\{\til{M}_m^h\}$ is a valid conventional complete measurement since $\til{U}_h$ is a conventional unitary.
\end{proof}

\begin{proposition}[\Dual extension and complex correction for measurements]\label{pr:measextcorr}
Consider $\{M_m^h\}$ a conventional measurement analytic in its real parameter $h$. We have $M_m^h = \hat{M}_m^h + O(h^2)$ and $p^h(m) = \hat{p}^h(m) + O(h^2)$.\\
Consider $\{M_m^\e\}$ a \dual measurement with Stinespring dilation $U_\e$. We have $M_m^h\big|_{\e=h} + O(h^2) = \til{M}_m^h$ and $p^\e\big|_{\e=h}(m) + O(h^2) = \til{p}^h(m)$.
\end{proposition}
\begin{proof}
First part, the Taylor expansion of $M_m^h$ gives
\begin{align*}
M_m^h &= M_m + h N_m + O(h^2) = \hat{M}_m^h + O(h^2).
\end{align*}
Hence, $(M_m^h)^\dagger M_m^h = (\hat{M}_m^h)^\dagger \hat{M}_m^h + O(h^2)$, so $p^h(m) = \hat{p}^h(m) + O(h^2)$.\\
Second part, by Prop.~\ref{pr:unitaryextcorr}, $\til{U}_h = U_\e\big|_{\e=h} + O(h^2)$ with $U_\e\big|_{\e=h} = U + hiHU$. Thus
\begin{align*}
\til{M}_m^h &= (\bra{m} \otimes I)\, \til{U}_h\, (\ket{0} \otimes I) = (\bra{m} \otimes I)(U_\e\big|_{\e=h} + O(h^2))(\ket{0} \otimes I) \\
&= (\bra{m} \otimes I) U_\e\big|_{\e=h} (\ket{0} \otimes I) + O(h^2) = M_m^h  + O(h^2).
\end{align*}
So $(\til{M}_m^h)^\dagger \til{M}_m^h = (M_m^h)^\dagger M_m^h + O(h^2)$, hence $\til{p}^h(m) = p^\e\big|_{\e=h}(m) + O(h^2)$.
\end{proof}
Example~\ref{ex:measextcorr} illustrates the use of this proposition.

\section{Application: the Dirac Quantum Walk}\label{sec:QCA}

In this section, we apply \dual quantum theory in order to provide a unified treatment of both 1/ the Dirac QW \cite{arrighi2020}---which is a fundamentally discrete space discrete time quantum algorithm simulating the propagation of a fermion and 2/ the Dirac equation---which is a fundamentally continuous space continuous time PDE describing the propagation of a fermion. Essentially, all we need to do is formulate the Dirac QW, i.e. take seriously the fact that $\e^2=0$.

The Dirac QW is a discrete model of a particle moving in discrete steps $\Delta_x=\Delta_t=\e$. It moves either left or right, according to its spin degree of freedom, but sometimes changes direction in a way that is proportional to $\e m$. Let $\ket{x+}$ be the basic state that represents a right-mover at $x$, and similarly let $\ket{x-}$ represent the left-mover. At any given time, the \dual states we employ are of the form $\ket{\psi(t)} = \sum_x \psi^+(x, t)\ket{x+}+\psi^-(x, t)\ket{x-}$. The relation between the amplitudes at step $t$ and at step $t+\e$ is given by
\begin{align}
 \begin{pmatrix}
  \psi^-(x-\e, t + \e) \\
  \psi^+(x+\e, t + \e)
 \end{pmatrix} = U_\e \begin{pmatrix}
  \psi^+(x, t) \\
  \psi^-(x, t)
 \end{pmatrix}
\end{align}\label{eq:applyingUe}
where $U_\e$ is the \dual evolution
\begin{equation}
        U_\e = \begin{pmatrix}
        -im\e & 1 \\
        1 & -im\e \\
\end{pmatrix}.
\end{equation}
In the literature, one usually finds
\begin{equation}
\til{U}_\e = \begin{pmatrix}
-i\sin(m\e) & \cos(m\e) \\
\cos(m\e) & -i\sin(m\e) \\
\end{pmatrix}
\end{equation}
as it feels more unitary. But, over the \dual numbers, $U_\e=\til{U}_\e$, and in \dual quantum theory $U_\e$ is already unitary.
These gates are arranged as in Fig. \ref{fig:diracqca}.
\begin{figure}[htp]
\centering
\includegraphics[width=0.6\linewidth]{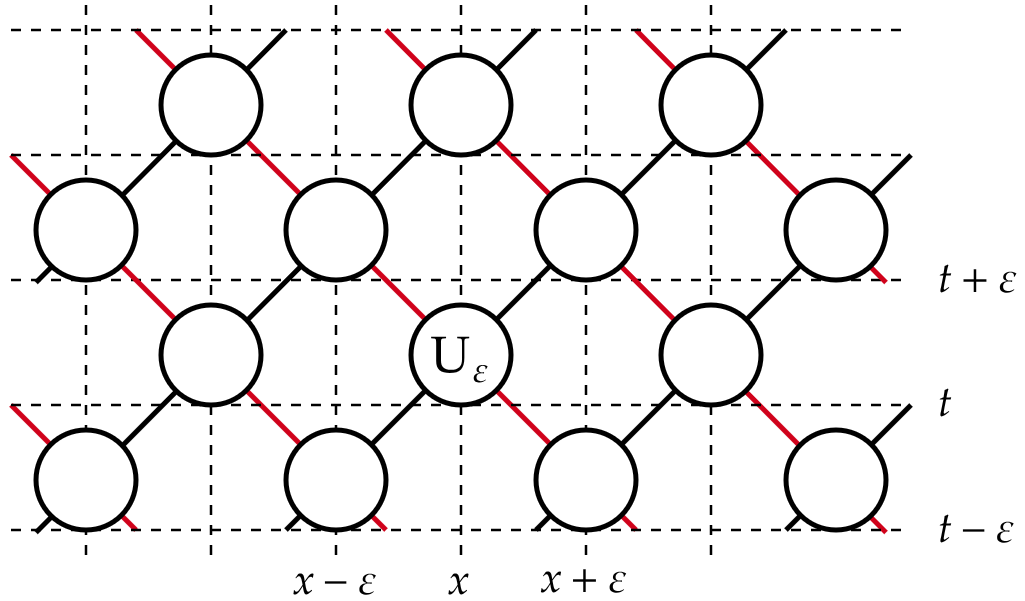}
\caption{The Dirac Quantum Walk}
\label{fig:diracqca}
\end{figure}
Applying the automatic differentiation Prop.~\ref{pr:auto}, to the second row of Eq.~\eqref{eq:applyingUe}, we get
\begin{align*}
  &\psi^+(x+\e, t+\e) = \psi^+(x, t) - im\e \psi^-(x, t) \\
  \Leftrightarrow\quad &\psi^+(x, t) + \e \frac{\partial}{\partial x}\psi^+(x, t) + \e \frac{\partial}{\partial t}\psi^+(x, t) = \psi^+(x, t) - im\e \psi^-(x, t) \\
  \Leftrightarrow\quad & \frac{\partial}{\partial t}\psi^+(x, t) = - \frac{\partial}{\partial x} \psi^+(x, t) - im \psi^-(x, t)
\end{align*}
Similarly, with the first row we get
\begin{align*}
   \frac{\partial}{\partial t}\psi^-(x, t) = \frac{\partial}{\partial x} \psi^-(x, t) - im \psi^+(x, t).
\end{align*}
These two PDEs can be stacked in a vector form, this is the Dirac equation:
\begin{align*}
 \partial_t \psi = - \sigma_3 \partial_x \psi - im \sigma_1 \psi
 \quad  \textrm{ with }\quad \psi(x, t) = \begin{pmatrix}
        \psi^+(x, t)\\
        \psi^-(x, t)\\
       \end{pmatrix}.
\end{align*}

In other words, in the \dual quantum theory formulation of the Dirac QW, the process of taking the continuum limit is fully automated. In fact, the very same object describes both the QW and the PDE, they are equated.

A major byproduct of this equality is that the continuous symmetries of the PDE are no longer broken by the QW formulation.

Indeed, the fundamental symmetry upon which the Dirac equation is built is Lorentz covariance, capturing the fact that Physics is unchanged under Lorentz transformations, which correspond to rescalings along the lightlike directions. For QWs, a precise notion of discrete Lorentz covariance exists and has been studied in detail in~\cite{ArrighiLorentzCovariance}. In that setting, a discrete Lorentz transform with integer parameters $\alpha,\beta$ is defined by replacing each spacetime point with a lightlike $\alpha\times\beta$ rectangular patch. Each such patch is obtained by applying isometric encodings $E_\alpha$ and $E_\beta$ on the left-moving and right-moving wires, then evolving under the same QW, although with a rescaled mass, throughout the patch. Performing such a discrete Lorentz transform and then taking the continuum limit, is the same as taking the continuum limit first and then performing the continuous Lorentz transform.
Discrete Lorentz covariance is the requirement that these patches join correctly, which amounts to the quantum circuit equality shown in Fig.~\ref{fig:CovRules}. I.e. there must exist a family of isometric encodings $E_\alpha$, $E_\beta$ such that for all $\alpha,\beta$, the QW evolution and the encoding commute in the sense of this figure.

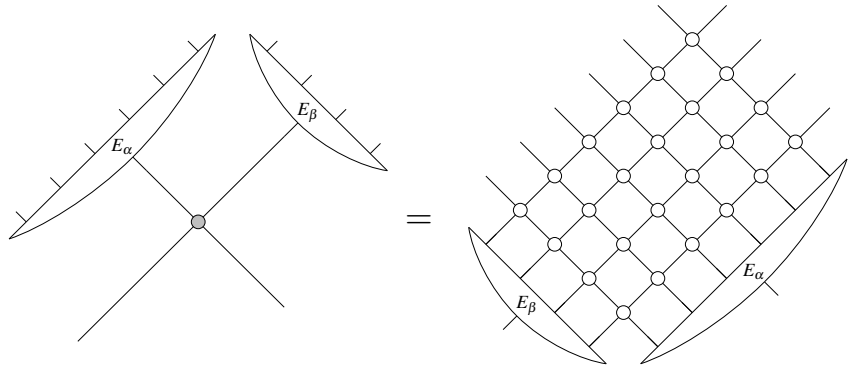
\begin{figure*}
\centering
\resizebox{0.7\linewidth}{!}{%
\begin{tikzpicture}[node distance=2cm,on grid,auto,/tikz/initial text=]
\path[use as bounding box] (-3.0cm,1cm) rectangle (14.5cm,8.5cm);

\newcommand{\gate}[3]{
\path (10 mm,05 mm) coordinate (bas);
\path (10 mm,\the\numexpr 5+#1*10 \relax mm) coordinate (haut);
\path (bas) edge (haut);
\path (4 mm, \the\numexpr 5+#1*5 \relax mm) coordinate (milieu);
\path (milieu)++(-4 mm, 0 mm) coordinate (milieul) ;
\path (milieu)++(0 mm,\the\numexpr #1*2 \relax mm) coordinate (controleh) ;
\path (milieu)++(03 mm, \the\numexpr #1*4 \relax mm) coordinate (controlehh) ;
\path (milieu)++(0,\the\numexpr -1*#1*2 \relax mm) coordinate (controleb) ;
\path (milieu)++(03 mm,\the\numexpr -1*#1*4 \relax mm ) coordinate (controlebb) ;
\path (milieu)++(3mm, 0mm) node[rotate=#2]{#3};
\draw[-] (bas).. controls (controlebb) and (controleb).. (milieu);
\draw[-] (milieu).. controls (controleh) and (controlehh) .. (haut);
\draw[-] (milieul) edge (milieu);
\foreach \i in {1,..., #1}{
   \path[-] (1 cm,\i cm) edge (1.3 cm,\i cm);
};
}

\newcommand{\mnwires}[7]{
\foreach \i in {1,..., #1}{
   \path[-] (1 cm,\i cm) edge (\the\numexpr 2+#2\relax cm,\i cm);
};
\gate{#1}{#3}{#5};
}

\newcommand{\crossing}[2]{
\mnwires{#1}{#2}{0}{}{$E_{\beta}$}{$E^\dagger_{\beta}$}{above right}
\savebox{\wires}{
\mnwires{#2}{#1}{-90}{}{$E_{\alpha}$}{$E^\dagger_{\alpha}$}{above left}
}
\node[rotate=90] at (\the\numexpr 2+#2\relax,-1) {\usebox{\wires}};
\foreach \i in {1,...,#2}{
 \foreach \j in {1,...,#1}{
\node[circle,draw,fill=white,inner sep=1mm] at (1+\i,\j) {};
}}
}

\savebox{\crossings}{
\begin{scope}[rotate=45]
\crossing{4}{6};
\end{scope}
}

\node[xshift=9cm]{\usebox{\crossings}};

\savebox{\crossings}{
\begin{scope}[rotate=45]
\gate{4}{0}{$E_{\beta}$};
\savebox{\gates}{\gate{6}{-90}{$E_{\alpha}$}}
\node[rotate=90] at (1,4) {\usebox{\gates}};
\path[-] (-6cm,2.5cm) edge (0cm,2.5cm);
\path[-] (-2.5cm,0cm) edge (-2.5cm,4cm);
\node[circle,draw,fill=lightgray,inner sep=1mm] at (-2.5,2.5) {};
\end{scope}
}

\node[xshift=4.5cm,yshift=4cm]{\usebox{\crossings}};
\node[xshift=5.5cm,yshift=4cm]{{\huge $=$}};
\end{tikzpicture}%
}
\caption{\label{fig:CovRules} Discrete Lorentz covariance as a quantum circuit equality.}
\end{figure*}

Unfortunately, the conventional quantum theory formulation of the Dirac QW is not exactly discrete Lorentz covariant. As shown by Fig.~\ref{fig:CovarianceFailure}, some terms in $O(\e^2)$ mess up the required circuit equality. Discrete Lorentz covariance holds only to first order, as illustrated in Fig.~\ref{fig:FirstOrderCovariance}.

\begin{figure*}
\centering
\begin{subfigure}[t]{0.38\textwidth}
\centering
\resizebox{\linewidth}{!}{%
\begin{tikzpicture}[auto,scale=4]
\path[use as bounding box] (0.5,0.5) rectangle (5,5);

\savebox{\patchuno}{
\begin{scope}[rotate=45]
\crossingswname{1}{1}{$\psi_+$}{$\dfrac{\psi_-}{\sqrt{2}}$}{black}{left}{right};
\end{scope}
\node at (-0.9,2.7) {$\left( 1-\dfrac{m'^2 \varepsilon^2}{2}\right)\dfrac{\psi_-}{\sqrt{2}}-\ii m' \varepsilon\psi_+$};
}

\savebox{\patchdos}{
\begin{scope}[rotate=45]
\crossingswname{1}{1}{$\left( 1-\dfrac{m'^2 \varepsilon^2}{2}\right)\psi_+-\dfrac{\ii m' \varepsilon\psi_-}{\sqrt{2}}$}{$\dfrac{\psi_-}{\sqrt{2}}$}{black}{left}{right};
\end{scope}
\node at (-1.6,2.6) {$\left( 1-\dfrac{3m'^2 \varepsilon^2}{2}\right)\dfrac{\psi_-}{\sqrt{2}}-\ii m' \varepsilon\psi_+$};
\node at (0.5,2.6) {$(1-m'^2 \varepsilon^2)\psi_+ - \dfrac{2\ii m' \varepsilon\psi_-}{\sqrt{2}}$};
}

\node at (2,0) {\usebox{\patchuno}};
\node at (3.6,1.6) {\usebox{\patchdos}};
\end{tikzpicture}%
}
\subcaption{\label{fig:CovarianceFailure} Failure of covariance at second order.}
\end{subfigure}\hfill
\begin{subfigure}[t]{0.6\textwidth}
\centering
\resizebox{\linewidth}{!}{%
\begin{tikzpicture}[auto,scale=3]
\path[use as bounding box] (-2.8,0.5) rectangle (7.1,7.5);

\savebox{\patchuno}{
\begin{scope}[rotate=45]
\crossingswname{1}{1}{$\dfrac{\psi_+}{\sqrt{2}}$}{$\dfrac{\psi_-}{\sqrt{3}}$}{black}{left}{right};
\end{scope}
}

\savebox{\patchdos}{
\begin{scope}[rotate=45]
\crossingswname{1}{1}{$\dfrac{\psi_+}{\sqrt{2}}-\dfrac{\ii m' \varepsilon\psi_-}{\sqrt{3}}$}{$\dfrac{\psi_-}{\sqrt{3}}$}{black}{left}{right};
\end{scope}
}

\savebox{\patchtres}{
\begin{scope}[rotate=45]
\crossingswname{1}{1}{$\dfrac{\psi_+}{\sqrt{2}}$}{$\qquad\dfrac{\psi_-}{\sqrt{3}}-\dfrac{\ii m' \varepsilon\psi_+}{\sqrt{2}}$}{black}{left}{left};
\end{scope}
\node at (-2,2.3) {${\color{blue}\dfrac{\psi_- (x-\varepsilon,t+\varepsilon)}{\sqrt{3}}} = \dfrac{\psi_- -\ii m \varepsilon \psi_+ }{\sqrt{3}} =$};
\node at (-1.68,1.92) {$=\dfrac{\psi_- }{\sqrt{3}}-\dfrac{2\ii m'\varepsilon \psi_+  }{\sqrt{2}}$};
}

\savebox{\patchcuatro}{
\begin{scope}[rotate=45]
\crossingswname{1}{1}{$\dfrac{\psi_+}{\sqrt{2}}-\dfrac{\ii m' \varepsilon\psi_-}{\sqrt{3}}$}{$\qquad\dfrac{\psi_-}{\sqrt{3}}-\dfrac{\ii m' \varepsilon\psi_+}{\sqrt{2}}$}{black}{left}{left};
\end{scope}
\node at (-1.3,2.6) {${\color{blue}\dfrac{\psi_- (x-\varepsilon,t+\varepsilon)}{\sqrt{3}}} = \dfrac{\psi_- }{\sqrt{3}}-\dfrac{2\ii m'\varepsilon \psi_+  }{\sqrt{2}}$};
\node at (3.8,3.8) {$ \dfrac{\psi_+ }{\sqrt{2}}-\dfrac{3\ii m'\varepsilon \psi_- }{\sqrt{3}}= \dfrac{\psi_+ -\ii m \varepsilon \psi_- }{\sqrt{2}}$};
\node at (4.2,3.5) {$ ={\color{violet} {\dfrac{\psi_+ (x+\varepsilon,t+\varepsilon)}{\sqrt{2}}}}$};
}

\savebox{\patchcinco}{
\begin{scope}[rotate=45]
\crossingswname{1}{1}{$\dfrac{\psi_+}{\sqrt{2}}-\dfrac{2\ii m' \varepsilon\psi_-}{\sqrt{3}}$}{$\dfrac{\psi_-}{\sqrt{3}}$}{black}{left}{right};
\end{scope}
\node at (1.1,2.0) {$ {\color{violet} {\dfrac{\psi_+ (x+\varepsilon,t+\varepsilon)}{\sqrt{2}}}}$};
}

\savebox{\patchseis}{
\begin{scope}[rotate=45]
\crossingswname{1}{1}{$\dfrac{\psi_+}{\sqrt{2}}-\dfrac{2\ii m' \varepsilon\psi_-}{\sqrt{3}}$}{$\qquad\dfrac{\psi_-}{\sqrt{3}}-\dfrac{\ii m' \varepsilon\psi_+}{\sqrt{2}}$}{black}{left}{left};
\end{scope}
\node at (-1.9,2.2) {${\color{blue}\dfrac{\psi_- (x-\varepsilon,t+\varepsilon)}{\sqrt{3}}} = \dfrac{\psi_- }{\sqrt{3}}-\dfrac{2\ii m'\varepsilon \psi_+  }{\sqrt{2}}$};
}

\node at (2,0) {\usebox{\patchuno}};
\node at (3.6,1.6) {\usebox{\patchdos}};
\node at (0.4,1.6) {\usebox{\patchtres}};
\node at (2.0,3.2) {\usebox{\patchcuatro}};
\node at (5.2,3.2) {\usebox{\patchcinco}};
\node at (3.6,4.8) {\usebox{\patchseis}};
\end{tikzpicture}%
}
\subcaption{\label{fig:FirstOrderCovariance} First-order covariance, hence exact \dual covariance.}
\end{subfigure}
\caption{{\em Discrete Lorentz covariance for the Dirac Quantum Walk.} Here the chosen isometries $E_\alpha$ and $E_\beta$ spread the incoming amplitude uniformly, and $m'=m/\sqrt{\alpha\beta}$.}
\end{figure*}
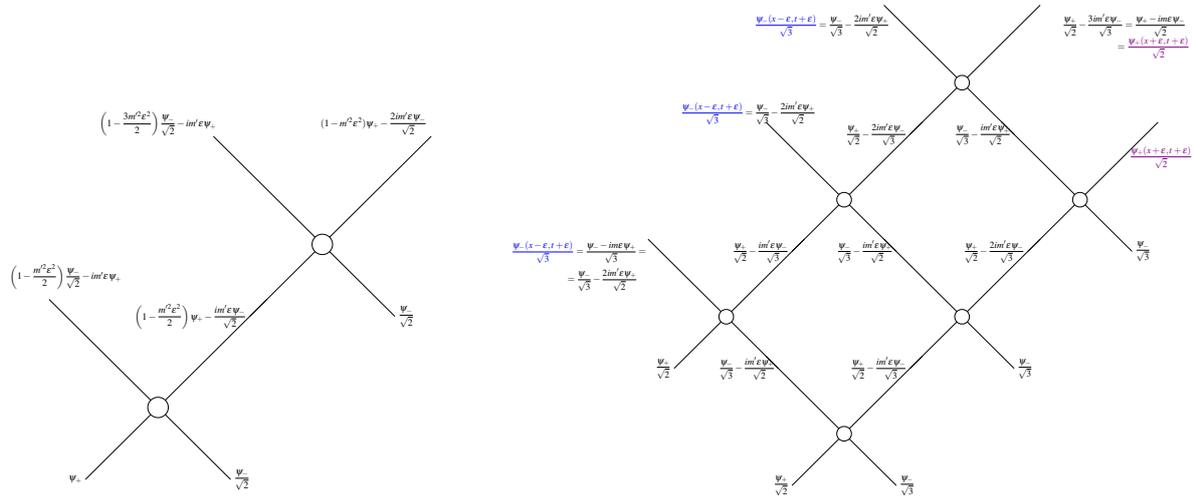

In the \dual quantum theory formulation of the Dirac QW, those second-order terms therefore vanish, thereby restoring exact discrete Lorentz covariance. This is a contribution, bearing in mind the many works dealing with discreteness induced Lorentz-symmetry breaking e.g. \cite{Rovelli2003reconcile,livine2004lorentz,dowker2004quantum,PaviaLORENTZ,PaviaLORENTZ2,DebbaschLORENTZ}.

Finally, notice that by applying the complex correction of Def.~\ref{def:unitaryextcorr}, we recover the conventional quantum theory Dirac QW $\til{U}_h$. Interestingly, this automates the exact path Meyer took, by {hand}, to design the first ever QW \cite{meyer1996quantum}. Prop.~\ref{pr:unitaryextcorr} ensures that $\til{U}_h=U_h+O(h^2)$.

\section{Conclusion}\label{sec:conclusion}

In this paper, we demonstrated that quantum theory can safely be extended from complex numbers to \dual numbers, and that this offers a common language in order to deal with both continuous and discrete physical models.

First, we advanced our knowledge of \dual linear algebra by demonstrating that \dual unitary operators are 1) diagonalizable and 2) related to \dual Hermitian \textit{via} exponentiation and logarithm maps. We also introduced a notion of ordering on \dual numbers.

Second, we presented a set of postulates extending quantum theory to $\mathbb{C}[\e]$ with $\e^2 = 0$. We showed that the equivalence between the discrete-time unitary formulation and the continuous-time Schrödinger equation continues to hold.

Third, we established the self-consistency of \dual quantum theory by showing norm and probability preservation, and introducing a translation between parametrized conventional quantum theory and \dual quantum theory. One can always be viewed as approximating the other up to $O(h^2)$.

Fourth, we applied this formalism to a concrete example, namely we were able to describe both the Dirac equation and its corresponding Discrete Time Quantum Walk, as a single object.

These results show the possibility and usefulness of applying \dual numbers to quantum theory. They are also of general mathematical interest as they contribute to the ongoing research on \dual linear operators.

We aim to pursue these results in order to reach a mathematical framework in which we can unify discrete and continuous Physics symmetries beyond Lorentz boosts.

Additionally, the foundational debate on the set of numbers to use for quantum mechanics could also benefit from a generalization of our extension to higher-order \dual numbers or general quaternions $H(\alpha, 0)$.

Another direction would be to explore connections with Grassmann numbers, which also obeys nilpotency conditions and may offer new ways of describing fermionic systems in a generalized \dual formalism of quantum mechanics.

{Ultimately, we wish to perfect the formalism in order to reconcile the following apparent paradox. On the one hand, quantum cellular automata seem to require $\e^2=0$ in order to be Lorentz covariant \cite{ArrighiLorentzCovariance}. But on the other hand, they seem to require $\e^2\neq 0$ in order to model particular phenomena, such as the electric term of QED \cite{ArrighiQED}. }

\paragraph*{Acknowledgements} This project was partially funded by the European Union through the MSCA SE project QCOMICAL, by the French National Research Agency (ANR): projects TaQC ANR-22-CE47-0012 and within the framework of `Plan France 2030', under the research projects EPIQ ANR-22-PETQ-0007, OQULUS ANR-23-PETQ-0013, HQI-Acquisition ANR-22-PNCQ-0001 and HQI-R\&D
ANR-22-PNCQ-0002, and by the WOST, WithOut SpaceTime project (https://withoutspacetime.org), grant ID\# 63683 from the John Templeton Foundation (JTF). The opinions expressed in this work are those of the author(s) and do not necessarily reflect the views of the John Templeton Foundation.

\bibliographystyle{eptcs}
\bibliography{bib}

\iftoappendix
\appendix

\section{Proofs for Sec.~\ref{sec:dualnumbers}}
\propexp*\label{app:prop:propexp}
\begin{proof}\label{app:proof:propexp}
\expandafter\the\csname savedproof_propexp_toks\endcsname
\end{proof}

\corroots*\label{app:cor:corroots}
\begin{proof}\label{app:proof:corroots}
\expandafter\the\csname savedproof_corroots_toks\endcsname
\end{proof}

\propauto*\label{app:prop:propauto}
\begin{proof}\label{app:proof:propauto}
\expandafter\the\csname savedproof_propauto_toks\endcsname
\end{proof}

\propdualorder*\label{app:prop:propdualorder}
\begin{proof}\label{app:proof:propdualorder}
\expandafter\the\csname savedproof_propdualorder_toks\endcsname
\end{proof}

\section{Proofs for Sec.~\ref{sec:linearalgebra}}

\propunitary*\label{app:prop:propunitary}
\begin{proof}\label{app:proof:propunitary}
\expandafter\the\csname savedproof_propunitary_toks\endcsname
\end{proof}

\propmatexp*\label{app:prop:propmatexp}
\begin{proof}\label{app:proof:propmatexp}
\expandafter\the\csname savedproof_propmatexp_toks\endcsname
\end{proof}

\propspecunit*\label{app:prop:propspecunit}
\begin{proof}\label{app:proof:propspecunit}
\expandafter\the\csname savedproof_propspecunit_toks\endcsname
\end{proof}

\prophermunit*\label{app:prop:prophermunit}
\begin{proof}\label{app:proof:prophermunit}
\expandafter\the\csname savedproof_prophermunit_toks\endcsname
\end{proof}

\propsemipos*\label{app:prop:propsemipos}
\begin{proof}\label{app:proof:propsemipos}
\expandafter\the\csname savedproof_propsemipos_toks\endcsname
\end{proof}


\fi

\section{Example for Sec.~\ref{sec:consistency}}

\begin{example}[Complex correction of a \dual measurement]\label{ex:measextcorr}
Consider the \dual measurement $\{M_m^\e\}$ with $M_m^\e = M_m + \e N_m$ given by:
\begin{multicols}{2}
\noindent
\begin{equation}
M_0 = \ketbra 0 0
\end{equation}
\noindent
\begin{equation}
M_1 = \ketbra 0 1
\end{equation}
\columnbreak
\noindent
\begin{equation}
N_0 = i \pi \ketbra 0 0
\end{equation}
\noindent
\begin{equation}
N_1 = i \pi (\frac 1 {\sqrt 3} \ketbra 1 0 + \ketbra 0 1 - \frac 2 {\sqrt 6} \ketbra 1 1)
\end{equation}
\end{multicols}

By Prop.~\ref{pr:dualstinespring} there exists $U_\e = U + \e iHU$ is a Stinespring dilation of $\{M_m^\e\}$. Let us construct such an $U_\e$. We need a conventional unitary $U$ and Hermitian $H$ such that $(\bra{m} \otimes I) U (\ket{0} \otimes I) = M_m$ and $(\bra{m} \otimes I)\, iHU\, (\ket{0} \otimes I) = N_m$.

The constraint $(\bra{m} \otimes I) U (\ket{0} \otimes I) = M_m$ is satisfied by the SWAP gate on the ancilla. In the chosen basis:
\begin{equation}
 U = \begin{pmatrix}
        1 & 0 & 0 & 0 \\
        0 & 0 & 1 & 0 \\
        0 & 1 & 0 & 0 \\
        0 & 0 & 0 & 1 \\
     \end{pmatrix}.
\end{equation}
The first two columns of $U^\dagger(iHU) = i U^\dagger H U$ are fixed by $N_m$; completeness of the dilation requires $-i U^\dagger(iHU) = U^\dagger H U$ to be Hermitian. One can complete the matrix accordingly; one solution is
\begin{equation}
        H = \pi \begin{pmatrix}
               1 & 0 & 0 & \frac{1}{\sqrt{3}} \\
               0 & 1 & 0 & 0 \\
               0 & 0 & 1 & -\frac{2}{\sqrt{6}} \\
               \frac{1}{\sqrt{3}} & 0 & -\frac{2}{\sqrt{6}} & 1 \\
            \end{pmatrix}
\end{equation}
We get $(\bra{m} \otimes I) U_\e (\ket{0} \otimes I) = M_m^\e$.

By Def.~\ref{def:measextcorr} the complex correction $\{M_m^\e\}$ is $\til{M}_m^h = (\bra{m} \otimes I)\, \til{U}_h\, (\ket{0} \otimes I)$ with $\til{U}_h = \exp(ihH)\, U$.
\begin{equation}
        \til{U}_h = \begin{pmatrix}
               \frac{z(\cos(\pi h)+2)}{3} & \frac{\sqrt{2}\,z\,(1-\cos(\pi h))}{3} & 0 & \frac{\sqrt{3}\,(z^2-1)}{6} \\[6pt]
               0 & 0 & z & 0 \\[6pt]
               \frac{\sqrt{2}\,z\,(1-\cos(\pi h))}{3} & \frac{z(2\cos(\pi h)+1)}{3} & 0 & \frac{\sqrt{6}\,(1-z^2)}{6} \\[6pt]
               \frac{\sqrt{3}\,(z^2-1)}{6} & \frac{\sqrt{6}\,(1-z^2)}{6} & 0 & \frac{z^2+1}{2}
            \end{pmatrix}
\end{equation}
where $z=e^{i\pi h}$. We get
\begin{subequations}
        \begin{align}
        \til{M}_0^h = \frac 1 3 \begin{pmatrix}
                z(\cos (\pi  h )+2) &  \sqrt{2}\,z\,(1-\cos(\pi h)) \\
               0 &  0
             \end{pmatrix}\\
        \til{M}_1^h = \frac 1 3 \begin{pmatrix}
                \sqrt{2}\,z\,(1-\cos(\pi h)) &  z(2\cos(\pi h)+1) \\
                i\sqrt{3}\,z\,\sin(\pi h) &  -i\sqrt{6}\,z\,\sin(\pi h)
             \end{pmatrix}
        \end{align}
\end{subequations}
Then $\{\til{M}_m^h\}$ is a valid conventional complete measurement. By Prop.~\ref{pr:measextcorr}, its effects and outcome probabilities agree up to $O(h^2)$ with those of $\{M_m^\e\}$.
\end{example}

\end{document}